\documentclass[12pt]{article}
\usepackage{amsthm,amssymb,amsmath}
\usepackage{enumerate}
\usepackage{color}
\usepackage{authblk}
\usepackage{algorithm}%
\usepackage{algpseudocode}
\usepackage{fullpage}

\newtheorem{theorem}{Theorem}[section]
\newtheorem{lemma}[theorem]{Lemma}
\newtheorem{proposition}[theorem]{Proposition}

\newtheorem{observation}[theorem]{Observation}
\newtheorem*{thm:main}{Theorem~\ref{thm:main}}
\usepackage{tikz}
\usetikzlibrary{matrix, calc, arrows}
\tikzstyle{every node}=[circle, draw, fill=black,
                        inner sep=0pt, minimum width=4pt]
\usepackage{hyperref}
\hypersetup{
    pdftitle=   {Deciding whether there are infinitely many prime graphs with forbidden induced subgraphs},
   pdfauthor=  {Robert Brignall, Hojin Choi, Jisu Jeong, Sang-il Oum}
}

\newcommand\abs[1]{\lvert #1\rvert}

\newcommand\ind{\preceq_i}

\newcommand\C{\mathcal C}

\newcommand\R{\mathcal R}
\newcommand\Free{\operatorname{Free}}
\providecommand{\keywords}[1]{\noindent\textbf{\textit{Keywords: }} #1}
\begin{document}

\title{Deciding whether there are infinitely many prime graphs with forbidden induced subgraphs}
\author[$\dagger$]{Robert Brignall}
\affil[$\dagger$]{School of Mathematics and Statistics

 The  Open University, Milton Keynes, UK

{\small robert.brignall@open.ac.uk}}
\author[$\ddagger$]{Hojin Choi}
\author[$\ddagger$]{Jisu Jeong}
\author[$\ddagger$]{Sang-il Oum\thanks{Supported by the National Research Foundation of Korea (NRF) grant funded by the Korea government (MSIT) (No. NRF-2017R1A2B4005020).}}
\affil[$\ddagger$]{Department of Mathematical Sciences

KAIST,  Daejeon, Korea

{\small hjchoi0330@gmail.com, jisujeong89@gmail.com, sangil@kaist.edu}
}

\date\today
\maketitle

\begin{abstract}
  A \emph{homogeneous set} of a graph $G$ is a set $X$ of vertices such that $2\le \lvert X\rvert <\lvert V(G)\rvert$ and no vertex in $V(G)-X$ has both a neighbor and a non-neighbor in $X$. A graph is \emph{prime} if it has no homogeneous set. 
  We present an algorithm to decide whether a class of graphs given by a finite set of forbidden induced subgraphs contains infinitely many non-isomorphic prime graphs.

\end{abstract}

\section{Introduction}\label{sec:intro}
All graphs in this paper are simple. 
We write $H \ind G$ if a  graph $H$ is isomorphic to an induced subgraph of a graph $G$, which is a subgraph of $G$ obtained by deleting some vertices. 
A class $\C$ of graphs is \emph{hereditary} if for all graphs $H$ and $G$, $H\in \C$ whenever $H\ind G$ and $G\in \C$.
For a set $X$ of graphs, we say $G$ is \emph{$X$-free} if $H\not\ind G$ for all $H\in X$. Let us write $\Free(X)$ to denote the class of $X$-free graphs. 
It is clear that for each hereditary class $\C$ of graphs, there exists a set $X$ of graphs such that $\C=\Free(X)$, simply by taking $X$ as $\ind$-minimal graphs not in $\C$.
Note that this set $X$ is not necessarily finite (for example, consider the class of forests, whose minimal forbidden set contains all cycles on three or more vertices).

 A \emph{homogeneous} set (also known in the literature as clans~\cite{EHR1997}, intervals~\cite{Ille1997, schmerl:critically-inde:}, or modules~\cite{Giakoumakis1997,Spinrad1992}) of a graph $G$ is a set $X$ of vertices such that $2\le \abs{X}<\abs{V(G)}$ and each vertex in $V(G)-X$ is either complete or anti-complete to $X$. A graph is \emph{prime}\footnote{Other terms that have been used for `prime' include indecomposable, irreducible, and primitive.} if it has no homogeneous set. 

For positive integers $n$, 
let $P_n$ be a path on $n$ vertices and 
let $K_{1,n}$ be a complete bipartite graph on $n+1$ vertices where one part consists of one vertex. 
In $P_4$-free graphs, also known as cographs~\cite{CLB1981}, it is
well known that they have no prime graphs on three or more vertices.
However,
in  $K_{1,3}$-free graphs, commonly known as claw-free graphs, 
we can easily find infinitely many prime graphs, such as all cycle
graphs on at least $5$ vertices.
Thus we may ask the following question: \emph{given a finite set $L$ of
graphs, can we decide whether there are infinitely many
non-ismorphic $L$-free prime graphs?}
We answer this question positively as follows.

\begin{theorem}\label{thm:main}
For a given finite set $L$ of graphs, 
there exists an algorithm to decide whether $\Free(L)$ contains infinitely many non-isomorphic prime graphs.
\end{theorem}

Prime graphs form the `building blocks' of all other graphs 
by means of the \emph{modular decomposition} (See~{\cite[Theorem 1.5.1]{survey}}).
The modular decomposition first appeared in the abstract of a talk by Fra\"iss\'e~\cite{fraisse:on-a-decomposit:} in 1953, although its first appearance in an article seems to be Gallai~\cite{gallai:transitiv-orien:}. It has since appeared in a number of contexts, ranging from game theory to combinatorial optimization. 

Theorem~\ref{thm:main} decides a property that has an interesting
consequences. if a hereditary class $\C=\Free(L)$ of graphs has only finitely many non-isomorphic prime graphs, then the class has a number of desirable properties. 
For example, $\C$ is well-quasi-ordered by the induced subgraph relation~{\cite[Theorem 6]{KL2011}}
(in other words, $\C$ contains no infinite set of graphs no one of which is an induced subgraph of any other), 
and every graph in $\C$ has bounded \emph{clique-width}~\cite{CO2000}, which itself gives rise to a number of desirable algorithmic properties, via the meta-theorem of Courcelle, Makowsky, and Rotics~\cite{courcelle:linear-time-sol:}.

Brignall, Ru{\v{s}}kuc, and Vatter~\cite{BRV2008} studied an analogous problem for permutations, under the `containment' ordering. In the theory of permutations, simple permutations correspond to prime graphs in our context. They proved that there exists an algorithm to determine whether a given hereditary class of permutations described by finitely many forbidden permutations admits infinitely many simple permutations. To prove the existence of a decision algorithm, they utilise a theorem on unavoidable subpermutations in large simple permutations by Brignall, Huczynska, and Vatter~\cite{brignall:decomposing-sim:}. 

For us, it is also necessary to understand unavoidable induced subgraphs in large prime graphs. Recently Chudnovsky, Kim, Oum, and Seymour~\cite{CKOS2015} proved such a theorem, which states that every sufficiently large prime graph contains one of a few large prime graphs as an induced subgraph. We will review this theorem in detail in Theorem~\ref{thm:CKOS}.  Our algorithm will check whether all these unavoidable induced subgraphs are forbidden by the given set $L$ of forbidden graphs. 
If all of them are forbidden, then $\Free(L)$ contains only finitely many non-isomorphic prime graphs and so the algorithm terminates with the answer NO. If at least one of them is not forbidden, then we prove that $\Free(L)$ contains arbitrarily large prime graphs and so the algorithm terminates with the answer YES.  

One outcome of Theorem~\ref{thm:CKOS} dominates the work to prove Theorem~\ref{thm:main}, namely the case of `chains' of length $n$, and this is covered in Section~\ref{sec:strings}. In theory, to handle this case one could employ automata-theoretic arguments analogous to those used in~\cite{BRV2008} to handle `pin sequences', the direct analogue of chains for permutations. Instead, we will present a purely combinatorial argument, using a few applications of the pigeonhole principle, to show that if a class $\Free(L)$ contains arbitrarily long chains, then it must contain arbitrarily long chains with a periodic construction, whose period is bounded by a function of the largest graph in $L$.

The remaining cases from Theorem~\ref{thm:CKOS} and hence the proof of Theorem~\ref{thm:main} are covered in Section~\ref{sec:main}.

\section{Unavoidable induced subgraphs in large prime graphs}\label{sec:unavoidable}

Chudnovsky, Kim, Oum, and Seymour~\cite{CKOS2015} proved that 
every sufficiently large prime graph contains one of a few large prime graphs as an induced subgraph.
After a couple of preliminary concepts, we introduce definitions of those large prime graphs and the result in this section. 

The \emph{$1$-subdivision} of a graph $G$ is the graph $H$ obtained from $G$ by subdividing every edge once. 
The \emph{line graph} of a graph $G$ is the graph $H$ whose vertex set is $V(H)=E(G)$ and two vertices $e_1$, $e_2$ are adjacent in $H$ if two edges $e_1$, $e_2$ share an end in $G$. We are particularly interested in the 1-subdivision of $K_{1,n}$, and the line graph of $K_{2,n}$, both of which are prime for all $n\geq 3$, and illustrated in Figure~\ref{fig:unavoidables}(i) and (ii), respectively.

The \emph{thin spider with $n$ legs} is the graph $H$ with vertex set $V(H)=\{v_1,v_2,\ldots,v_n\}\cup\{u_1,u_2,\ldots,u_n\}$
and edge set $E(H)=\{v_iu_i:1\le i \le n\}\cup\{u_iu_j: 1\le i < j \le n\}$.
The \emph{half-graph of height $n$} is the graph $H$ with vertex set $V(H)=\{v_1,v_2,\ldots,v_n\}\cup\{u_1,u_2,\ldots,u_n\}$
and edge set $E(H)=\{v_iu_j: 1\le i \le j \le n\}$.
The graph $H'_{n,I}$ has vertex set $V(H'_{n,I})=\{v_1,v_2,\ldots,v_n\}\cup\{u_1,u_2,\ldots,u_n\}\cup\{w\}$
and edge set $E(H'_{n,I})=\{wv_i:1\le i \le n\}\cup\{v_iu_j: 1\le i \le j \le n\}\cup\{u_iu_j: 1\le i < j \le n\}$.
Finally, the graph $H^*_n$ has vertex set $V(H^*_n)=\{v_1,v_2,\ldots,v_n\}\cup\{u_1,u_2,\ldots,u_n\}\cup\{w\}$
and edge set $E(H^*_n)=\{wv_1\}\cup\{v_iu_j: 1\le i \le j \le n\}\cup\{u_iu_j: 1\le i < j \le n\}$.
Examples of these graphs are illustrated in Figure~\ref{fig:unavoidables}(iii)--(vi), and it is easy to see that these graphs are prime.\label{graphs-all-prime}

A \emph{chain} $C$ of length $n$ is a sequence $v_0,v_1,\ldots,v_n$ of distinct vertices such that for each $i\in\{2,\ldots,n\}$, $v_i$ is adjacent to all $v_0,v_1,\ldots,v_{i-2}$ but not $v_{i-1}$, or 
non-adjacent to all $v_0,v_1,\ldots,v_{i-2}$ but adjacent to $v_{i-1}$.
We call $v_0$ the \emph{first vertex} of the chain. 
The graph induced by a chain of length $n$ is prime, or is prime after discarding one of the vertices $v_0$ or $v_1$, as shown by the following result.

\begin{proposition}[{\cite[Corollary 2.3]{CKOS2015}}]\label{prop:chain}
Every chain of length $n>3$ contains a chain of
length $n-1$ inducing a prime graph.\end{proposition}

Note that, in a slight departure from~\cite{CKOS2015}, we will not necessarily require that a chain is contained inside some specified graph. Instead, chains can be considered as sequences of vertices, which may or may not be embedded inside some larger graph, depending on the context. Additionally, we may from time to time abuse notation by referring to the chain when we mean the graph induced by a chain.

We are now ready to state the main result of~\cite{CKOS2015}, which provides the structural basis for our algorithm.

\begin{theorem}[Chudnovsky, Kim, Oum and Seymour~\cite{CKOS2015}]\label{thm:CKOS}
For every integer $n\ge 3$, there exists $N$ such that 
every prime graph with at least $N$ vertices contains one of the following graphs as an induced subgraph.
\begin{enumerate}[(i)]
\item The $1$-subdivision of $K_{1,n}$ or its complement.
\item The line graph of $K_{2,n}$ or its complement.
\item The thin spider with $n$ legs or its complement.
\item The half-graph of height $n$.
\item The graph $H'_{n,I}$.
\item The graph $H^*_n$ or its complement.
\item A prime graph induced by a chain of length $n$.
\end{enumerate}
\end{theorem}

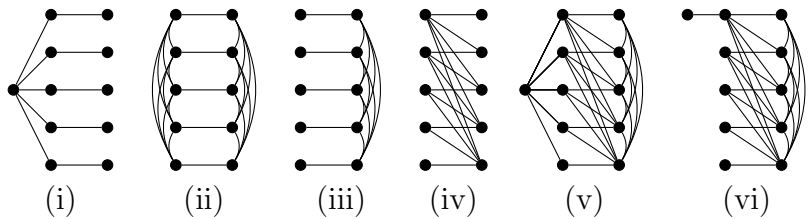
\begin{figure}
\centering
\begin{tabular}{cccccc}
\begin{tikzpicture}[scale=0.5]
\node (core) at (0,3) {};
\foreach \j in {5,...,1}
{
	\node (\j1) at (1,\j) {};
	\node (\j2) at (2.5,\j) {};
	\draw (core) -- (\j1) -- (\j2);
}
\end{tikzpicture}
&
\begin{tikzpicture}[scale=0.5]
\foreach \j in {1,...,5}
{
	\node (\j1) at (1,\j) {};
	\node (\j2) at (2.5,\j) {};
	\draw (\j1) -- (\j2);
}
\foreach \j/\k in {1/2,1/3,1/4,1/5,2/3,2/4,2/5,3/4,3/5,4/5}
{
	\draw[bend left=30] (\k2) to (\j2);
	\draw[bend right=30] (\k1) to (\j1);
}
\end{tikzpicture}
&
\begin{tikzpicture}[scale=0.5]
\foreach \j in {1,...,5}
{
	\node (\j1) at (1,\j) {};
	\node (\j2) at (2.5,\j) {};
	\draw (\j1) -- (\j2);
}
\foreach \j/\k in {1/2,1/3,1/4,1/5,2/3,2/4,2/5,3/4,3/5,4/5}
	\draw[bend left=30] (\k2) to (\j2);
\end{tikzpicture}
&
\begin{tikzpicture}[scale=0.5]
\foreach \j in {5,...,1}
{
	\node (\j1) at (1,\j) {};
	\node (\j2) at (2.5,\j) {};
	\foreach \k in {\j,...,5}
		\draw (\k1) -- (\j2);
}
\end{tikzpicture}
&
\begin{tikzpicture}[scale=0.5]
\node (core) at (0,3) {};
\foreach \j in {5,...,1}
{
	\node (\j1) at (1,\j) {};
	\node (\j2) at (2.5,\j) {};
	\foreach \k in {\j,...,5}
		\draw (core) -- (\k1) -- (\j2);
}
\foreach \j/\k in {1/2,1/3,1/4,1/5,2/3,2/4,2/5,3/4,3/5,4/5}
	\draw[bend left=30] (\k2) to (\j2);
\end{tikzpicture}
&
\begin{tikzpicture}[scale=0.5]
\node (core) at (0,5) {};
\foreach \j in {5,...,1}
{
	\node (\j1) at (1,\j) {};
	\node (\j2) at (2.5,\j) {};
	\foreach \k in {\j,...,5}
		\draw (\k1) -- (\j2);
}
\draw (core) -- (51);
\foreach \j/\k in {1/2,1/3,1/4,1/5,2/3,2/4,2/5,3/4,3/5,4/5}
	\draw[bend left=30] (\k2) to (\j2);
\end{tikzpicture}\\
(i)&(ii)&(iii)&(iv)&(v)&(vi)
\end{tabular}
\caption{Examples of the unavoidable graphs of cases (i)--(vi) in Theorem~\ref{thm:CKOS}.}
\label{fig:unavoidables}
\end{figure}

Note that in the characterization of Theorem~\ref{thm:CKOS}, the complements of the half-graphs (case (iv)) and $H'_{n,I}$ (case (v)) both contain (as induced subgraphs) graphs of the same type, with two vertices removed. 
Since the graphs in cases (i)--(vi) of Theorem~\ref{thm:CKOS} admit regular structures,
it is relatively straightforward to check whether a class $\Free(L)$ contains arbitrarily large ones.
The details are provided in Section~\ref{sec:main}.

\section{Chains and strings}\label{sec:strings}

In this section, we consider the chains that arise in case (vii) of Theorem~\ref{thm:CKOS}. Note that the complement of a chain is again a chain. 

For convenience, we seek to describe an encoding of chains as strings over the alphabet $\{0,1\}$. First, we introduce some elementary concepts about strings.

A \emph{$(0,1)$-string} (or simply a \emph{string}) is an element of $\{0,1\}^*$, 
where $\{0,1\}^*$ is the set of all finite sequences of $0$ and $1$. 
The \emph{length} of a string $S$ is the number of $0$'s and $1$'s in the string and is denoted by $\abs{S}$.
Given strings $S$ and $T$, we denote the \emph{concatenation} (defined in the natural way) by $ST$. For example, if $S=011$ and $T=101$, then $ST=011101$. 
Let $S^t$ denote the concatenation of $t$ copies of a string $S$. For example, $S^3=SSS$. 

We say that $T$ is a \emph{factor} of $S$, or $S$ \emph{contains} $T$ as a factor, if there exist strings $X$ and $Y$ such that $S=XTY$. An \emph{occurrence} of $T$ in $S$ is a pair $(T,i)$ such that $S=XTY$ and $\abs{X}=i-1$ 
(that is, $T$ is a factor of $S$ that starts at the $i$-th letter). Furthermore, we say that the occurrences $(S_1,i_1)$ and $(S_2,i_2)$ of two (possibly equal) factors inside some string $S$ with $i_1\leq i_2$ are \emph{1-disjoint} if $i_1+\abs{S_1} < i_2$ (in other words, there is at least one letter of $S$ that is not used in either of the occurrences, but lies `between' $S_1$ and $S_2$), and they \emph{intersect} if $i_1+\abs{S_1} > i_2$.

We are now ready for the basic encoding of chains into strings, which we will denote by $\phi$. 
For a chain $C=v_0,v_1,\ldots,v_k$ of length $k$, let $\phi(C)=s_1s_2\cdots s_k$ where $s_i=0$ if $v_i$ is adjacent to $v_{i-1}$, and $s_i=1$ otherwise for each $i\in\{1,\ldots, k\}$. 
Note that $\phi$ is a bijection between chains and strings, but recall that the graphs induced by two distinct chains $C_1$ and $C_2$ can be isomorphic and so a graph that is induced by some chain does not necessarily have a unique representation as a string. Note also that if $C$ contains $k+1$ vertices, then $\phi(C)$ contains $k$ letters, because the first vertex is not assigned a letter. See Figure~\ref{fig:chains}.
\newcommand{\chainstring}[2][]{
	\def\binseq{#2}
	\node (0) at (0,0) {};
	\foreach \j [count=\i, remember=\j as \jj (initially 0)] in \binseq {
		\ifthenelse{\j=2}{%
			\node[fill=none,draw=none] at (\i,-1) {$#1$};
		}{%
			\node[fill=none,draw=none] at (\i,-1) {\j};
			\node (\i) at (\i,0) {};
			\pgfmathtruncatemacro{\im}{\i-1};
			\ifthenelse{\j=0}{%
				\ifthenelse{\jj=2}{}{\draw (\i) -- ({\i-1},0);}
			}{%
				\ifthenelse{\i>1}{\draw[bend left=80] (0) to (\i);}{}
				\foreach \x [count=\xx] in \binseq {
					\ifthenelse{\xx<\im \AND \x<2}{%
						\draw[bend left=80] (\xx) to (\i);}{}%
				}
			}
		}
	}
}
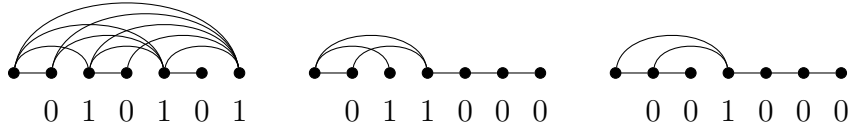
\begin{figure}
\centering
\begin{tikzpicture}[scale=0.5]
\chainstring{0,1,0,1,0,1}
\begin{scope}[shift={(8,0)}]
	\chainstring{0,1,1,0,0,0}
\end{scope}
\begin{scope}[shift={(16,0)}]
	\chainstring{0,0,1,0,0,0}
\end{scope}
\end{tikzpicture}
\caption{Examples of chains, and their encodings as strings via the bijection $\phi$. Note the two examples on the right give rise to graphs that are isomorphic.}
\label{fig:chains}
\end{figure}

We say that a graph $G$ is \emph{induced by} a string $S$ if $G$ is induced by $C = \phi^{-1}(S)$. Similarly, we say that a string $S$ \emph{contains} a graph $G$ if the graph induced by $S$ contains $G$ as an induced subgraph. 

In addition to encoding chains into strings, we also need to be able to encode subgraphs of strings, in order to identify when a given string contains graphs from the minimal forbidden set $L$. To this end, suppose that $G$ is a graph on $n$ vertices that embeds inside some string $S$. If $\phi^{-1}(S)=v_0,v_1,\ldots,v_k$, then $G$ is isomorphic to the graph induced on the subsequence $v_{i_1},v_{i_2},\ldots, v_{i_{n}}$ that corresponds to the chosen embedding, where $0\leq i_1 < i_2 <\cdots <i_{n} \leq k$. We now define a new encoding $\psi$ from subsequences of chains (or embeddings of graphs into chains) into strings over the three-letter alphabet $\{0,1,\mid\,\}$.

For each $j$ ranging from 2 to $n$, the encoding $\psi$ writes symbols according to the following rules: if $i_j = i_{j-1}+1$, then write $0$ if $v_{i_j}$ is adjacent to $v_{i_{j-1}}$, and 1 otherwise. When $i_j > i_{j-1}+1$, write $\mid 0$ if $v_{i_j}$ is \emph{not} adjacent to $v_{i_{j-1}}$ (and all earlier vertices), and $\mid 1$ otherwise. 
If $G$ is isomorphic to the graph induced on the subsequence $M$ of some chain,
then we call $\psi(M)$ a \emph{representation} of $G$. 
A \emph{block} of a representation is a maximal factor that contains only the letters 0 and 1. See Figure~\ref{fig:subchains}. Note that if a representation begins with the symbol $\mid$, then we will assume that there is an empty block preceding it.

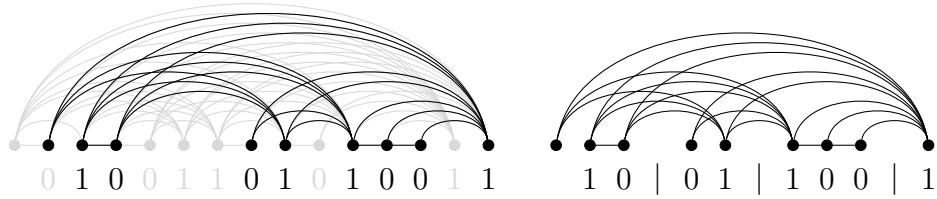
\begin{figure}
\centering
\begin{tikzpicture}[scale=0.45]
\begin{scope}[shift={(3,0)}]
\chainstring[\mid]{1,0,2,0,1,2,1,0,0,2,1}
\end{scope}
\begin{scope}[gray!30,every node/.append style={fill=gray!30}, shift={(-13,0)}]
\chainstring{0,1,0,0,1,1,0,1,0,1,0,0,1,1}
\end{scope}
\begin{scope}[shift={(-12,0)}]
\chainstring{1,0,2,2,2,0,1,2,1,0,0,2,1}
\end{scope}
\end{tikzpicture}
\caption{On the left, an embedding $M$ of a graph inside the chain with string $01001101010011$. On the right, the encoding of $M$ is the representation $\psi(M)=10\mid01\mid100\mid1$ with blocks 10, 01, 100 and 1.}
\label{fig:subchains}
\end{figure}

Before we go further, we need to make a few remarks about the strings created under the encoding $\psi$. Suppose that $M$ is an embedding (or subsequence) of a graph on $n$ vertices inside some chain.
\begin{itemize}
\item $\psi(M)$ has exactly $n-1$ symbols that are 0 or 1.
\item $\psi(M)$ cannot contain the factor $\mid\,\mid$, nor can it end with the symbol~$\mid$. Therefore, there are at most $n-1$ instances of the symbol $\mid$ in $\psi(M)$. 
\item $\psi(M)$ therefore contains at most $2n-2$ letters.
\item $\psi$ is not a bijection, because it does not remember the specific positions of vertices of $M$ in the chain.
\end{itemize} 

At this point we make an important observation: one can view the reverse process of $\psi$, from words over the alphabet $\{0,1,\mid\,\}$ to graphs, as a monadic second-order transduction, from which it is possible to conclude that the edge relation of subgraphs of chains is definable by a monadic second-order formula. This gives rise to a decision procedure for whether $\Free(L)$ contains arbitrarily long or not via the Backwards Translation Theorem (see \cite[Theorem 7.10]{courcelle:graph-structure:}), as the language over $\{0,1,\mid\,\}$ corresponding to $\Free(L)$ is regular. This approach is essentially the same as the one given in the case of permutations, see~\cite{BRV2008}, but it is not the approach we use here.

Instead, the decision procedure we present here comprises two parts and is elementary (in that it requires only the pigeonhole principle applied to the structures introduced in this section so far).
First, we establish that if there exists a chain of a specified (large) length in $\Free(L)$, 
then there exists arbitrarily long chains with a periodic structure, where the size of the period is bounded above by a function of the largest forbidden graph in $L$ (this may be compared to the `pumping lemma' in the study of regular languages). Note that by exhaustively checking membership in $\Free(L)$ of all chains of the specified large length, this result is already sufficient for a decision procedure. However, the second part of our procedure gives us a simpler method, namely a check for whether a particular chain sequence can be repeated arbitrarily often.

Now we consider the total number of possible representations of graphs on $n$ vertices.
Each representation is obtained from a $(0,1)$-string of length $n-1$
by inserting at most $n-1$ copies of the symbol $\mid$. 
There are $2^{n-1}$ $(0,1)$-strings of length $n-1$, and there are
$2^{n-1}$ choices of inserting the symbol $\mid$ or not at each position.
Thus,
we deduce the following observation.

\begin{observation}\label{obs:numrep}
For each positive integer $n$, 
there are at most $2^{2n-2}$ representations of graphs on $n$ vertices. %
Moreover, each such representation $R$ has at most $n$ blocks.
\end{observation}

Now consider a representation $R$ of a graph with $n$ vertices. Although we cannot recover the specific embedding of this graph in a chain that gave rise to the representation, we can reconstruct the graph from $R$ in the natural way: create vertices $v_1,v_2,\dots v_n$, where $v_2,\dots,v_{n}$ correspond to the \mbox{non-$\mid$} symbols in $R$, reading from left to right. For each $i$ ($2\leq i\leq n$), the adjacencies of $v_i$ to the previous $i-1$ vertices is determined by the letter of $R$ corresponding to $v_i$ (which is either 0 or 1), and the letter (if it exists) immediately preceding this one in $R$ (specifically, whether this symbol is $\mid$ or not). 

Given the above reconstruction process, each representation $R$ corresponds to a unique graph $G$. However, each graph $G$ can have several corresponding representations -- see Figure~\ref{fig:multireps} for an example. We let $\R_G$ denote the set of all representations that correspond to a given graph $G$. 
Note that %
$\abs{\R_G}\le 2^{2\abs{V(G)}-2}$ by Observation~\ref{obs:numrep}.

\begin{figure}
\centering
\begin{tikzpicture}[scale=0.45]
\begin{scope}
\chainstring[\mid]{0,2,1,1}
\end{scope}
\begin{scope}[shift={(6,0)}]
\chainstring[\mid]{2,1,2,1,1}
\end{scope}
\begin{scope}[shift={(13,0)}]
\chainstring[\mid]{1,2,1,2,1}
\end{scope}
\begin{scope}[shift={(20,0)}]
\chainstring[\mid]{2,0,2,1,2,1}
\end{scope}
\end{tikzpicture}
\caption{Four different representations of the same graph $K_4-e$.}
\label{fig:multireps}
\end{figure}
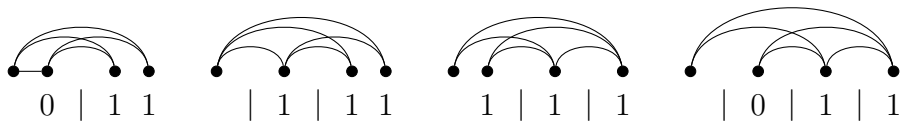

Our final preparatory task is to observe how a representation $R\in\R_G$ can be embedded in some given $(0,1)$-strings $S$. 
We say that the string $S$ \emph{contains} the representation $R$ if 
\begin{enumerate}[(1)]
\item each block of $R$ is embedded as a factor in $S$, and 
\item every pair of distinct blocks $B_i$ and $B_j$ are embedded as 1-disjoint factors, with the factor corresponding to $B_i$ preceding that of $B_j$ if and only if $B_i$ precedes $B_j$ in $R$.
\end{enumerate}

Now, we introduce two lemmas for the proof of Theorem~\ref{thm:main}. 

\begin{lemma}\label{lem:lem3}
Let $L$ be a set of graphs having at most $n$ vertices. 
If there exists a $(0,1)$-string $T$ of length at least $\lceil\frac{(n-1)4^n+1}{3}\rceil(2^{n-2}+n-1)$ containing no graphs in $L$, %
then there exists a $(0,1)$-string $S$ of length at most $2^n$ such that the $(0,1)$-string $S^k$ contains no graph in $L$ for all $k$. %
\end{lemma}

\begin{proof}
Let $\R=\cup_{G\in L} \R_G$ be the set of all representations of
graphs from $L$. Note that, by Observation~\ref{obs:numrep}, we have
$\abs{\R} \leq 
\sum_{k=1}^n 2^{2k-2}\le
4^n/3$. Furthermore, each representation in $\R$ has at most $n$ blocks. 

Let $s=\lceil\frac{(n-1)4^n+1}{3}\rceil$. 
We may assume that $\abs{T}=s(2^{n-2}+n-1)$. 
We can rewrite $T=T_1\ell_1T_2\ell_2\cdots T_s\ell_s$ where $|T_i|=2^{n-2}+n-2$ and $\ell_i = 0$ or $1$ for all $i$. Thus, the $T_i$ are pairwise 1-disjoint.
We claim that 
there exists $j^*$ such that for every representation $R\in\R$, at least one block of $R$ is not a factor of $T_{j^*}$. 
Suppose not.
Then for each $j\in\{1,2,\ldots, s\}$, there exists $R_j\in\R$ 
such that $T_j$ contains each block of the representation $R_j$ as a factor. 
Note that the blocks of $R_j$ in $T_j$ may overlap and may appear in the incorrect order. 
Since $s>(n-1)4^n/3$ and $\abs{\R}\leq 4^n/3$, by the pigeonhole principle, at least $n$ of the $T_j$ must contain all the blocks from one particular representation $R^*\in\R$ as a factor. That is, there exists a subsequence $j_1,j_2,\ldots,j_n$ of $1,2,\dots,s$ such that $R_{j_1} = R_{j_2} = \cdots = R_{j_n}=R^*$.
This means that by considering the factor of $T_{j_1}$ equal to the first block of $R^*$, the factor of $T_{j_2}$ equal to the second block, and so on, and recalling that the $T_{j_k}$ are pairwise 1-disjoint, we find that $T$ contains the representation $R^*$. Therefore $T$ contains some $G\in L$, a contradiction which proves the claim.

Now, $T_{j^*}$ does not contain at least one block of every representation $R\in\R$ as a factor. By the pigeonhole principle, since $\abs{T_{j^*}}=2^{n-2}+n-2$, there exist two (not necessarily disjoint) occurrences $(A,a_1)$ and $(A,a_2)$ in $T_{j^*}$ such that $\abs{A}=n-2$, and $a_1<a_2$. That is, we find the same factor of $n-2$ letters occurring at least twice in $T_{j^*}$. %

Now, consider the occurrence $(S,a_1)$ in $T_{j^*}$ where $S$ is a factor of $T_{j^*}$ of length $a_2-a_1$, 
in other words, $T_{j^*}=K_1SAK_2$ for some (possibly empty) prefix $K_1$ and suffix $K_2$ of $T_{j^*}$. 
Note that $\abs{S}\leq 2^{n-2}$ since $T_{j^*}$ has length $2^{n-2}+n-2$ and $\abs{A}=n-2$. 
We claim that $\phi^{-1}(S^k)\in\Free(L)$ for all $k$. 

Suppose to the contrary that there exists $k$ such that $S^k$ contains some representation $R\in\R$. 
By construction of $T_{j^*}$, there is some block $B$ of $R$ that is not contained in $T_{j^*}$ as a factor, and therefore $B$ is not contained in $SA$ or in $S$ as a factor. 
Moreover, by construction of $S$, we observe that either $S^k$ is a factor of $SA$, or $SA$ is a factor of $S^k$. 
See Figure~\ref{fig:sa}. 

\begin{figure}[h]
	\begin{center}
		\includegraphics[scale=0.8]{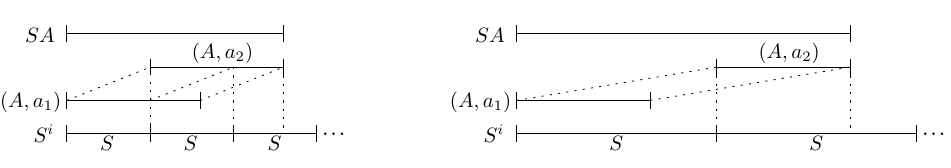}
	\end{center}
  \caption{Since both occurrences $(A,a_1)$ and $(A,a_2)$ represent the same factor, we can deduce that for each positive integer $i$, either $S^i$ is a factor of $SA$ or $SA$ is a factor of $S^i$.}
  \label{fig:sa}
\end{figure}

If $S^k$ is a factor of $SA$, then since the block $B$ is a factor of $S^k$, it is also a factor of $SA$, which is a contradiction. Therefore, $SA$ is a factor of $S^k$. We may assume that $B$ is embedded as a factor in $S^k$ starting from an entry in the first copy of $S$. Note that $A$ contains precisely $n-2$ letters, and $B$ contains at most $n-1$ letters. From this, we conclude that $B$ embeds into $SA$ (starting from an entry in the prefix $S$), another contradiction.

Thus we conclude that $S^k$ contains no representation $R\in \R$ for all $k$, which completes the proof.
\end{proof}

Lemma~\ref{lem:lem3} tells us that if a class $\Free(L)$ contains arbitrarily long chains then it contains arbitrarily long chains with a periodic construction, whose period is at most $2^n$. Our next lemma gives us the necessary practical condition for our decision procedure to test whether a string can be repeated arbitrarily many times or not.

\begin{lemma}\label{lem:lem2}
Let $L$ be a set of graphs having at most $n$ vertices.
Let $S$ be a string.
If $S^{2n-1}$ contains none of the graphs in $L$, %
then $\Free(L)$ contains $\phi^{-1}(S^k)$ for all $k$.
\end{lemma}

\begin{proof}
Suppose that the lemma is false. Let $M$ be the minimum number such that $S^M$ contains at least one graph $G\in L$. 
This means that there exists a representation $R\in\R_G$ which is contained in $S^M$. Fix one such embedding of $R$ in $S^M$. 
Since $M\geq 2n$ and  $\abs{V(G)}\leq n$, 
there exist two consecutive copies of $S$ neither of which is used in the embedding of $R$ in $S^M$. We can therefore eliminate one of these two copies of $S$ while still ensuring that the blocks of $R$ are $1$-disjoint (to ensure $R$ can still be embedded in the resulting string). That is, $S^{M-1}$ still contains $R$,
which is a contradiction since $M$ is the minimum number such that $S^M$ contains at least one graph in $L$. %
\end{proof}

\section{Proof of the main result}\label{sec:main}

In this section, we prove our main result, Theorem~\ref{thm:main}. Recall the statement of our main theorem.

\begin{thm:main}
For a given finite set $L$ of graphs, 
there exists an algorithm to decide whether $\Free(L)$ contains infinitely many non-isomorphic prime graphs.
\end{thm:main}

Let $\mathcal{G}_n$ be the set that consists of 
the $1$-subdivision of $K_{1,n}$ and its complement, 
the line graph of $K_{2,n}$ and its complement, 
the thin spider with $n$ legs and its complement, 
the half-graph of height $n$, 
the graph $H'_{n,I}$, and
the graph $H^*_n$  and its complement. %
In other words, $\mathcal{G}_n$ contains one representative of each
type of graph in Theorem~\ref{thm:CKOS} except for chains. Note that
it is routine to check for each $n$ that all the graphs in $\mathcal{G}_{n}$ are prime.

By Theorem~\ref{thm:CKOS}, a large prime graph that does not contain a chain of length $n$ must contain a graph in $\mathcal{G}_n$.
For a graph in $\mathcal{G}_n$, 
it is easy to deduce the following lemma by the definition of $\mathcal{G}_n$. 
For an example, suppose that a graph $G$ with $n$ vertices is an induced subgraph of the $1$-subdivision of $K_{1,N+1}$. 
Let $v$ be a vertex of degree $N+1$ in the $1$-subdivision of $K_{1,N+1}$, let $u_1, u_2, \ldots, u_{N+1}$ be neighbors of $v$, and let $v_i$ be a neighbor of $u_i$ other than $v$ for each $i$. 
Now, there exist $u_i$ and $v_i$ such that neither $u_i$ nor $v_i$ are in $G$. We delete $u_i$ and $v_i$ from the $1$-subdivision of $K_{1,N+1}$ to obtain the $1$-subdivision of $K_{1,N}$ that contains $G$ as an induced subgraph.  
We can prove similarly for other cases. 

\begin{lemma}\label{lem:case}
Let $G$ be a graph on $n$ vertices and let $N$ be an integer with $N\geq \max{\{n,3\}}$.
If $G$ is an induced subgraph of some graph in $\mathcal{G}_{N+1}$,
then there exists a graph $H$ in $\mathcal{G}_N$ such that $G$ is an induced subgraph of $H$.\qed
\end{lemma}

Finally, we give the proof of our main theorem, providing Algorithm~\ref{alg:main}.

\begin{proof}[Proof of Theorem~\ref{thm:main}]
Let $n\geq 3$ be the minimum integer such that every graph in $L$ has at most $n$ vertices. 
By the contrapositive statement to Lemma~\ref{lem:case}, if $\mathcal{G}_{n}$ has a graph in $\Free(L)$ then 
$\Free(L)$ must contain a graph from $\mathcal{G}_N$ for every $N\geq n$. Hence $\Free(L)$ has infinitely many non-isomorphic prime graphs. 

Now, we may assume that every graph in $\mathcal{G}_{n}$ is not in $\Free(L)$. 
By Theorem~\ref{thm:CKOS}, 
it is enough to decide whether $\Free(L)$ has infinitely many non-isomorphic prime graphs induced by chains. 
If there exists a string $S$ of length at most $2^n$ such that $\phi^{-1}(S^{2n-1})\in\Free(L)$, 
then by Proposition~\ref{prop:chain} and Lemma~\ref{lem:lem2}, 
$\Free(L)$ has infinitely many non-isomorphic prime graphs induced by chains. 

On the other hand, if $\phi^{-1}(S^{2n-1})\not\in\Free(L)$ for every string $S$ of length at most $2^n$, %
then by Lemma~\ref{lem:lem3} the maximum length of a chain contained in $\Free(L)$ is less than $\lceil\frac{(n-1)4^n+1}{3}\rceil(2^{n-2}+n-1)$, which implies that $\Free(L)$ has only finitely many non-isomorphic prime graphs. 
\end{proof}

\begin{algorithm}
  \caption{Does $\Free(L)$ contain infinitely many prime graphs?} \label{alg:main}
\begin{algorithmic}[1]
\State{Let $L$ be the input set of graphs and let $n\geq 3$ be the minimum integer such that every graph in $L$ has at most $n$ vertices.}
\If{$\mathcal{G}_{n}$ has a graph in $\Free(L)$}
\State{output YES.}
\ElsIf{there exists a string $S$ of length at most $2^n$ such that the string $S^{2n-1}$ contains no graph in $L$}
\State{output YES.}
\Else
\State{output NO.}
\EndIf
\end{algorithmic}
\end{algorithm}

\section{Concluding remarks}

\paragraph{Complexity of the procedure.} We have not made any particular effort to optimize the procedure described above. The majority of the work lies in determining whether a hereditary class $\Free(L)$ admits arbitrarily long chains or not, and here one may need to exhaust over all $2^{2^n+1}-1$ chains of length at most $2^{n}$, where $n = \max_{G\in L}|G|$. 
By contrast, Lemma~\ref{lem:case} shows that in order to check whether $\Free(L)$ contains arbitrarily large prime graphs of the other types listed in Theorem~\ref{thm:CKOS}, it suffices to check whether each of the 10 graphs in $\mathcal{G}_n$ (each having at most $2n+1$ vertices) contains some graph in $L$.

In the analogous problem of deciding whether a hereditary class of permutations contains only finitely many simple permutations, a recent paper due to Bassino, Bouvel, Pierrot and Rossin~\cite{bassino:a-polynomial-al:b} establishes an algorithm with run time $O(nk\log (nk) + n^{2k})$, where $n$ is the size of the largest forbidden permutation, and $k$ is the number of forbidden permutations. It is quite possible that a similar detailed analysis of chains in graphs could lead to a much more efficient algorithm.

\paragraph{Finding all the prime graphs in a class.} If our decision procedure returns YES, then in theory it could provide a `certificate' of an infinite family of prime graphs that the class contains. On the other hand, if the procedure returns NO, then Lemmas~\ref{lem:lem3} and~\ref{lem:case} give bounds on the number of vertices that the largest prime graph in the class can contain. However, the following result (recently re-discovered by Chudnovsky and Seymour~\cite{chudnosvky:growing-without:}), gives a more practical method that may terminate sooner:

\begin{proposition}[Schmerl and Trotter~\cite{schmerl:critically-inde:}]
Let $n\ge 3$ be an integer. Every prime graph on $n$ vertices contains a prime induced subgraph on $n-1$ or $n-2$ vertices.
\end{proposition}
Furthermore, the only prime graphs that do not contain a prime graph on 1 fewer vertices are the half-graphs of height $n$, and their complements. Thus, to list all prime graphs in a class, one can successively generate and check for membership the prime graphs of each order, and halt as soon as one finds two consecutive integers where the hereditary class contains no prime graphs of that order. 

\paragraph{Classes with infinitely many minimal forbidden graphs}
One may ask whether it is possible for a hereditary class $\C=\Free(L)$ to contain only finitely many prime graphs when $L$ is an \emph{infinite} minimal set of forbidden graphs. The answer to this is no: any hereditary class containing only finitely many prime graphs possesses the property of being \emph{labelled} well-quasi-ordered (see~\cite[Theorem 2]{atminas:labelled-induced:}), and any such class is defined by a finite set of minimal forbidden graphs (this latter observation is essentially due to Pouzet~\cite{pouzet:un-bel-ordre-da:}). 

The same observation (that a hereditary class with only finitely many prime graphs is defined by finitely many minimal forbidden graphs) also leads to a quick proof of a special case of the results concerning `prime extensions': namely that a finite set of prime graphs necessarily only has finitely many prime extensions (see Giakoumakis and Olariu~\cite{GO2007}). 

\paragraph*{Acknowledgment}
The authors would like to thank the reviewer for
providing a careful review and an alternative proof of the theorem
in terms of monadic second-order logic.


\end{document}